\documentclass[journal]{IEEEtran}
\ifCLASSINFOpdf
\else
\fi
\usepackage{pgfplots}
\usepackage{amsmath}
\usepackage[utf8]{inputenc}
\usepackage[english]{babel}
\usepackage{graphicx}
\usepackage{tikz}
\usepackage{array}
\usepackage{booktabs}
\usepackage{caption}
\usepackage{tabu}
\usetikzlibrary{shapes,arrows,shadows}
\usepackage{epstopdf}
\usepackage{amsthm}

\newtheorem{lemma}{Lemma}

\usepackage{algorithm}
\usepackage{algpseudocode}
\usepackage{pifont}
\usepackage{subfig}
\usepackage{amssymb}

\theoremstyle{definition}

\makeatletter
\def\BState{\State\hskip-\ALG@thistlm}
\makeatother

\begin{document}
\title{Biased Contribution Index: A Simpler Mechanism  to Maintain  Fairness in Peer to Peer  Network }
%\author{Author 1~and~Author 2} 

\author{Sateesh Kumar Awasthi~and~Yatindra~Nath~Singh, \IEEEmembership{Senior~Member,~IEEE}% <-this % stops a space%
\thanks{Sateesh Kumar Awasthi  is with the Department
of Electrical Engineering, Indian Institute of Technology, Kanpur, India, e-mail: (sateesh@iitk.ac.in)}% <-this % stops a space
\thanks{Yatindra~Nath~Singh  is with the Department
of Electrical Engineering, Indian Institute of Technology, Kanpur, India, e-mail: (ynsingh@iitk.ac.in)}}% <-this % stops a space
\maketitle

\begin{abstract}
To maintain fairness, in the terms of resources shared by an individual peer, a proper incentive policy is required in a peer to peer network. This letter proposes, a simpler mechanism to rank the peers based on their resource contributions to the network. This  mechanism  will  suppress the free riders from downloading the resources from the network. Contributions of the peers are biased in such a way that  it can balance the download and upload amount of resources at each peer. This mechanism can be implemented in a distributed system and it converges much faster than the other  existing approaches. 
\end{abstract}


% \fi
\IEEEpeerreviewmaketitle

\section{Introduction}
\IEEEPARstart{T}{he} peers are motivated to share the resources in a peer to peer network, if they get at least the amount of data, what they have uploaded to the network. Most ideal situation is that, when upload and download amount for each peer is same. We call this, a fair situation in a peer to peer network. Of course, there will be no free riders in this situation. Therefore some  mechanism is required to assure this. In this letter, we are proposing a simple mechanism called biased contribution index(BCI).\par
To achieve the aforementioned state, many incentive mechanisms have been studied \cite{global}, \cite{shared}, \cite{tit}, \cite{mtit}, \cite{give}. Among these,   global approaches \cite{global}, \cite{shared} performed better compared to local approaches \cite{tit}, \cite{mtit}, \cite{give}. In global approaches, shared history of peers   in entire network is taken into consideration. It gives the wider view of peer's cooperation in the network. However, global approaches are not trivial to implement in the network. Like in \cite{shared}, each peer needs to keep the record of every transaction history, regardless of whether he was directly involved in them or not. In \cite{global}, peer's contribution largely depends upon the contribution of peer with whom  it is transacting so each  peer needs to make some rough estimate of the impact on its contribution before each transaction. Our approach is similar to \cite{global}, but  equal importance is given to the actual contribution of peer and the contribution of peer with whom it is transacting. It is simpler to implement in the network due to its faster convergence. With the suitable numerical example we compared the convergence result of our approach with \cite{global}. We found that our approach performs far better in all the cases.\par
Rest of this paper is organized as follows. Section \ref{section1} presents the network model and introduction of biased contribution index.  Solution of biased contribution index is given in section \ref{solution}. Analysis of proposed algorithm is discussed in section \ref{analysis}. Section \ref{result} shows the numerical results, and in section \ref{con}, conclusion is presented.


\section{Network model and biased contribution Index }\label{section1}
Let there be $N$ peers in a peer to peer network. Peers share their resources with each other and their contribution is calculated globally. A simple metric which can best reflect the contribution of the peers in the network, could be the ratio of its total upload to the network to the total download from the network. But to motivate the peers to upload  more to the peers contributing more and to download more from the peers contributing less, we need to bias this ratio by some incentive factor. Let this incentive factor be, $x_i$, for peer $i$.  For peer $i$, we define, the upload to download ratio(biased ratio)as
\[R_i= \mathbf{\frac{e_i.s.x}{e_i.s^T.x}}.\]

%\begin{figure}\label{fig1}
%\begin{tikzpicture}
%\begin{axis}[
%    axis lines = left,
%    xlabel = $R_i$,
%    ylabel = {$x_i=f(R_i)$},
%]
%Below the red parabola is defined
%\addplot [
%    domain=0:5, 
%    samples=100, 
%    color=red,
%]
%{x/(1+x)};
%\addlegendentry{$x^2 - 2x - 1$}
 
%\end{axis}
%\end{tikzpicture}
%\caption{Incentive factor $x_i$ as a function of bias ratio $R_i$}
%\end{figure}

Here, $\mathbf{s}$ is $NXN$ share matrix. Its $ij^{th}$ element represents the  amount of resource shared by peer $i$ to peer $j$. $\mathbf{e_i}$ is the row vector with its $i^{th}$ entry as '1' and all other entries as  zero. The $\mathbf{x}$ is a vector containing incentive factors. Let us define the incentive factor as a monotonically increasing function of biased ratio, i.e.
\[x_i= \frac{R_i}{1+R_i},\]
\[=\frac{\mathbf{e_i.s.x}}{\mathbf{e_i.s.x + e_i. s^T.x}}.\]
We call this incentive factor as  biased contribution index (BCI). \par Now to start the process of sharing we need to give some initial value of BCI to all the peers. Let us define a parameter $\alpha \in (0, 1)$ to decide the initial value of BCI. Later we will see that the parameter $\alpha$ is also related to the speed of convergence. The biased contribution index is modified to include this parameter $\alpha$, and is given by.
\begin{equation}\label{equ1}
x_i= \begin{cases}
\alpha \frac{\mathbf{e_i.s.x}}{\mathbf{e_i.s.x + e_i. s^T.x}} + (1-\alpha), &\text{if $\mathbf{e_i.s.x + e_i. s^T.x} \neq 0$}.\\
(1-\alpha/2), &{otherwise}.
\end{cases}
\end{equation}
Here, $(1-\alpha/2)$ is the initial value of biased contribution index, when neither upload nor download has happened at the node. In the network of $N$ nodes, there will be $N$ unknowns and $N$ nonlinear equations. We will see in the next section that, these equations can be solved by a suitable iterative function.\par 
   Peers are allowed to take the resources from network only if their biased contribution index is  above a certain threshold value. Therefore  every peer will try to increase its biased contribution index, so that it can get the required amount of resources whenever needed.\par It can be observed easily from equation \ref{equ1} that a peer's biased contribution index will be higher if\\
 1). Its contribution $s_{ij}$ is higher,\\
 2). It shares more of its resources with higher contributing peer,  and\\
 3). It takes more of the services from lower contributing peer.\\
 Therefore, intuitively we can say that this metric can assure fairness in the whole network. Later in section \ref{analysis}, we will give a mathematical justification for this.

%\IEEEpubidadjcol

\section{Solution of Biased Contribution index}\label{solution}
The BCI of any peer is expressed in the terms of the BCI of other peers. If $\mathbf{e_i.s.x + e_i. s^T.x} \neq 0$, for $i=1,2....N$, then for $N$ peers network, the equation \ref{equ1},  can be expressed in the form of matrix.
\[\mathbf{x}=\mathbf{diag}[d_1,d_2,...d_N].\mathbf{s.x} + (1-\alpha)\mathbf{e}\]
Here $d_i= \alpha/(\mathbf{e_i.s.x + e_i. s^T.x})$ and $\mathbf{e}$ is vector with each element as '1'. We  propose following Lemmas in this regard. 
\begin{lemma}\label{lemma0}
The biased contribution index vector $\mathbf{x} \in [(1-\alpha), 1]^N$
\end{lemma}
\begin{proof}
When any peer $i$ only takes the resources from the network and does not contribute any thing, then $\mathbf{e_i.s.x=0}$ and $\mathbf{e_i.s^T.x \neq 0}$. In this case,  biased contribution index of peer $i$ will be minimum and it will be,
\[x_i=\alpha.0 + (1-\alpha)=(1-\alpha).\]
When a peer $i$ only contribute the resources to the network without taking any thing, then $\mathbf{e_i.s.x \neq 0}$ and $\mathbf{e_i.s^T.x = 0}$. In this case,  biased contribution index of peer $i$ will be maximum and it will be,
\[x_i=\alpha.1 + (1-\alpha)=1.\]
In all other cases it will be in between these values. Hence $\mathbf{x} \in [(1-\alpha), 1]^N.$
\end{proof}
\begin{lemma}\label{lemma1}
Let $\mathbf{s}$ be $NXN$ non negative, irreducible matrix, then $\mathbf{x}$ in the above expression can be calculated by the iterative function
\[\mathbf{x^k}=\mathbf{\phi(x^{k-1})},\] where $i^{th}$ element of iterative function  $\mathbf{\phi(x^{k-1})}$ is $\alpha [\mathbf{e_i.s.x^{k-1}}/{\mathbf{(e_i.s.x^{k-1} + e_i. s^T.x^{k-1})}}] + (1-\alpha)$.
\end{lemma} 
\begin{proof}
 The $i^{th}$ element of iterative function $\mathbf{\phi(x^{k-1})}$  is 
 \[x_i^{k} = \alpha \frac{\mathbf{e_i.s.x^{k-1}}}{\mathbf{(e_i.s.x^{k-1} + e_i. s^T.x^{k-1})}} + (1-\alpha).\]              
 Let $x_i^k$ and $x_i^{k-1}$ are, far from actual solution $x_i$ by  $\delta x_i^k$ and $\delta x_i^{k-1}$ respectively, then  
 \[x_i + \delta x_i^{k} = \alpha \frac{\mathbf{e_i.s.(x + \delta x^{k-1})}}{\mathbf{[e_i.s.(x + \delta x^{k-1}) + e_i. s^T.(x + \delta x^{k-1})]}} + (1-\alpha).\]      Let $\mathbf{s+s^T =s^{'}},$ then
   \[x_i + \delta x_i^{k} = \alpha \frac{\mathbf{e_i.s.(x + \delta x^{k-1})}}{\mathbf{e_i.s^{'}.(x + \delta x^{k-1})}} + (1-\alpha).\]
   \[\delta x_i^{k} = \alpha \frac{\mathbf{e_i.s.(x + \delta x^{k-1})}}{\mathbf{e_i.s^{'}.(x + \delta x^{k-1})}} + (1-\alpha)- x_i\]  Using equation \ref{equ1},
\[\delta x_i^{k} = \alpha \frac{\mathbf{e_i.s.(x + \delta x^{k-1})}}{\mathbf{e_i.s^{'}.(x + \delta x^{k-1})}} - \alpha \frac{\mathbf{e_i.s.x}}{\mathbf{e_i.s^{'}.x }}.\] 
\[\delta x_i^{k} =\alpha \frac{\mathbf{e_i.s.x}}{\mathbf{e_i.s^{'}.x }}\Bigg[\frac{\big(1+ \frac{\mathbf{e_i.s.\delta x^{k-1}}}{\mathbf{e_i.s^.x }}\big)}{\big(1+\frac{\mathbf{e_i.s^{'}.\delta x^{k-1}}}{\mathbf{e_i.s^{'}.x }}\big)}-1 \Bigg].\] It can be observed from equation \ref{equ1}, 
\[x_i > \alpha \frac{\mathbf{e_i.s.x}}{\mathbf{e_i.s.x + e_i. s^T.x}}=\alpha \frac{\mathbf{e_i.s.x}}{\mathbf{e_i.s^{'}.x }};\] hence,
\[\delta x_i^{k} < x_i\Bigg[\frac{\big(1+ \frac{\mathbf{e_i.s.\delta x^{k-1}}}{\mathbf{e_i.s^.x }}\big)}{\big(1+ \frac{\mathbf{e_i.s^{'}.\delta x^{k-1}}}{\mathbf{e_i.s^{'}.x }}\big)}-1 \Bigg]\]
\[= \frac{x_i}{\big(1+\frac{\mathbf{e_i.s^{'}.\delta x^{k-1}}}{\mathbf{e_i.s^{'}.x }}\big)}\bigg[{\big(1+ \frac{\mathbf{e_i.s.\delta x^{k-1}}}{\mathbf{e_i.s^.x }}\big)}-{\big(1+\frac{\mathbf{e_i.s^{'}.\delta x^{k-1}}}{\mathbf{e_i.s^{'}.x }}\big)}\bigg]\]
 \[= \frac{x_i}{\big(1+\frac{\mathbf{e_i.s^{'}.\delta x^{k-1}}}{\mathbf{e_i.s^{'}.x }}\big)}\bigg[\bigg({\frac{\mathbf{e_i.s}}{\mathbf{e_i.s^.x }}\bigg)}-{\bigg(\frac{\mathbf{e_i.s^{'}}}{\mathbf{e_i.s^{'}.x }}\bigg)}\bigg]\mathbf{\delta x^{k-1}}\]
\[= \frac{1}{\big(1+\frac{\mathbf{e_i.s^{'}.\delta x^{k-1}}}{\mathbf{e_i.s^{'}.x }}\big)}\bigg[\bigg({\frac{x_i\mathbf{e_i.s}}{\mathbf{e_i.s^.x }}\bigg)}-{\bigg(\frac{x_i\mathbf{e_i.s^{'}}}{\mathbf{e_i.s^{'}.x }}\bigg)}\bigg]\mathbf{\delta x^{k-1}}\]
\[= f_i(\mathbf{\delta x^{k-1}})[\mathbf{A_i-B_i}]\mathbf{\delta x^{k-1}}\]
Where $\mathbf{A_i}$ and $\mathbf{B_i}$ are $i^{th}$ row of $NXN$ matrix $\mathbf{A}$ and $\mathbf{B}$ respectively. It can be observed about matrix $\mathbf{A}$ and $\mathbf{B}$ that, $\mathbf{Ax=x}$ and
$\mathbf{Bx=x}$.\par
 Matrix $\mathbf{A}$ and $\mathbf{B}$ are derived from matrix $\mathbf{s}$. Since matrix $\mathbf{s}$ is irreducible hence matrix $\mathbf{A}$ and $\mathbf{B}$ will also be irreducible. Elements of vector $\mathbf{x}$ are positive (see Lemma \ref{lemma0}), so for non negative matrix $\mathbf{s}$, matrix $\mathbf{A}$ and $\mathbf{B}$ will also be non negative. Therefore spectral radius of matrix $\mathbf{A}$ and $\mathbf{B}$ will be '1' and corresponding  eigen vector will be $\mathbf{x}$ (see \cite{nonnegative}).\par If $\mathbf{\delta x^{k-1} << x},$ then $f_i(\mathbf{\delta x^{k-1}} )\approx 1$. Hence,
 \[\mathbf{\delta x^k} < [\mathbf{A-B}]\mathbf{\delta x^{k-1}} < [\mathbf{A-B}]^{k}\mathbf{\delta x^{0}} \] 
 \[lim_{k\to\infty}\mathbf{\delta x^k} < lim_{k\to\infty}[\mathbf{A-B}]^k\mathbf{\delta x^0}=\mathbf{0}\] (see Theorem 1 in \cite{abs})
 
And if $\mathbf{\delta x^{k-1} > x}$, then $f_i(\mathbf{\delta x^{k-1}} ) < 1$. Hence in this case, $\delta x_i^{k-1}$ will decrease more rapidly till $\mathbf{\delta x^{k-1} << x}.$
 
 \par Hence $\mathbf{x}$ can be calculated by the aforementioned iterative function.
\end{proof}
\begin{lemma}\label{lemma2}
If $\mathbf{x}=a\mathbf{e}$ is the solution of biased contribution index then  $a$ will be $(1-\alpha/2)$. 
\end{lemma}
\begin{proof}
If $\mathbf{e_i.s.x + e_i. s^T.x} \neq 0$, equation \ref{equ1} can be written as 
\[x_i(\mathbf{e_i.s.x + e_i. s^T.x}) = \alpha\mathbf{e_i.s.x}+ (1-\alpha)(\mathbf{e_i.s.x + e_i. s^T.x})\] 
\[\Rightarrow x_i(\mathbf{e_i.s.x + e_i. s^T.x}) = \mathbf{e_i.s.x}+ (1-\alpha)\mathbf{e_i. s^T.x}\] for $i=1,2,...N$. Hence, above relation can be written in the form of matrix as follows.
\[ \mathbf{diag(x)(s+s^T)x=sx} + (1-\alpha)\mathbf{s^Tx}.\]
Here $\mathbf{diag(x)}$ is $NXN$ diagonal matrix with its $ii^{th}$ element as $x_i$. Now if $\mathbf{x}=a\mathbf{e}$ then
\[a\mathbf{I(s+s^T)}a\mathbf{e}= a\mathbf{s.e} +(1-\alpha)a \mathbf{s^T.e}\] 
\[\Rightarrow a^2\mathbf{(s+s^T)}\mathbf{e}= a(\mathbf{s.e} +(1-\alpha) \mathbf{s^T.e})\] Pre-multiply by $\mathbf{e^T}$ on both side
\[a^2\mathbf{e^T(s+s^T)}\mathbf{e}= a(\mathbf{e^T.s.e} +(1-\alpha)\mathbf{e^T.s^T.e})\] 
\[\Rightarrow a^2\mathbf{(e^T.s.e+e^T.s^T.e)} = a(\mathbf{e^T.s.e} +(1-\alpha)\mathbf{e^T.s^T.e})\] for any matrix $\mathbf{s}$, $\mathbf{e^T.s.e}$ will be the sum of all of  its elements. Hence $\mathbf{e^T.s.e}=\mathbf{e^T.s^T.e}=T$, and above expression can be written as \[a^2(T+T)=a(T+(1-\alpha)T)\] since  $a\in [(1-\alpha), 1]$,  hence \[a(2T)=(2-\alpha)T\] since  $T \neq 0$, hence $a=(1-\alpha/2)$
 \end{proof}
\begin{lemma}\label{lemma3}
If $\mathbf{x}=(1-\alpha/2)\mathbf{e}$ is the solution of biased contribution index then  \[\mathbf{(s^T-s)e_i^{T}} \perp \mathbf{e}  \hspace{5mm}\forall i.\]
\end{lemma}
\begin{proof}
Substituting $\mathbf{x}=(1-\alpha/2)\mathbf{e}$ in equation \ref{equ1}
\[(1-\alpha/2)=\alpha \frac{(1-\alpha/2)\mathbf{e_i.s.e}}{(1-\alpha/2)(\mathbf{e_i.s.e + e_i. s^T.e})} + (1-\alpha)\]
\[\Rightarrow \hspace{5mm} \alpha/2=\alpha \frac{\mathbf{e_i.s.e}}{(\mathbf{e_i.s.e + e_i. s^T.e})} \hspace{24mm}\]
\[\Rightarrow \hspace{5mm} \alpha(\mathbf{e_i.s.e + e_i. s^T.e})=2\alpha\mathbf{e_i.s.e} \hspace{19mm}\]
\[\Rightarrow \hspace{5mm} \alpha \mathbf{e_i(s-s^T).e = 0} \hspace{37mm}\] 
\[\Rightarrow \hspace{5mm}\alpha \mathbf{((s^T-s)e_i^T)^T.e = 0} \hspace{31mm}\] 
Since $\alpha \neq 0$, hence
\[\mathbf{(s^T-s)e_i^{T}} \perp \mathbf{e}  \hspace{5mm}\forall i\]
\end{proof}
\begin{lemma}\label{lemma4}
If $\mathbf{se=s^Te}$, then the biased contribution index vector, $\mathbf{x}=(1-\alpha/2)\mathbf{e}.$
\end{lemma}
\begin{proof}
$\mathbf{se=s^Te \Rightarrow e^Ts^T=e^Ts}$, now from equation \ref{equ1}, if $\mathbf{e_i.s.x + e_i. s^T.x} \neq 0$.\\
\[ \mathbf{diag(x)(s+s^T)x=sx} + (1-\alpha)\mathbf{s^Tx}\] Pre-multiply  by $\mathbf{e^T} $ on both side
\[ \mathbf{e^Tdiag(x)(s+s^T)x=e^Tsx} + (1-\alpha)\mathbf{e^Ts^Tx}\] 
\[\Rightarrow \mathbf{x^T(s+s^T)x=e^Tsx} -(\alpha/2)\mathbf{e^Ts^Tx}+ (1-\alpha/2)\mathbf{e^Ts^Tx}\]
\[\Rightarrow \mathbf{x^T(s+s^T)x=e^Tsx} -(\alpha/2)\mathbf{e^Tsx}+ (1-\alpha/2)\mathbf{e^Ts^Tx}\]
\[\Rightarrow \mathbf{x^T(s+s^T)x}=(1-\alpha/2)\mathbf{e^Tsx}+ (1-\alpha/2)\mathbf{e^Ts^Tx}\hspace{9mm}\]
\[\Rightarrow \mathbf{x^T(s+s^T)x}=(1-\alpha/2)\mathbf{e^T(s+s^T)x} \hspace{26mm}\]
\[\Rightarrow [\mathbf{x^T}-(1-\alpha/2)\mathbf{e^T}]\mathbf{(s+s^T)x=0} \hspace{36mm}\]It is clear that $(\mathbf{s+s^T})$ is non negative matrix and $\mathbf{x}>0$ hence $\mathbf{(s+s^T)x \neq 0}$. Hence
$\mathbf{x}=(1-\alpha/2)\mathbf{e}.$

\end{proof}
\section{Analysis of biased contribution index}\label{analysis}

\subsection{Solution Of Free Riding and Collusion}
Initially every peer is allowed to take some resources from the network; otherwise process of sharing will not start. Hence, initial BCI of $(1-\alpha/2)$ for each peer is justified. But as soon as  BCI is updated, free rider's BCI will reach at minimum level. Because for any free rider $i$, $\mathbf{e_i.s.x}=0$, hence from equation \ref{equ1},  $x_i=(1-\alpha)$. This will disqualify them from  taking any resources from the network in future, until they acquire  sufficient BCI.\par The contributing peer will always gain the BCI and resource taking peer will always loose the BCI. Hence, peers will always avoid reporting the false transaction and thus, collusion can be avoided in the network.
\subsection{Justification For Fairness}
We can observe the following from the above discussion.\\
1) If the biased contribution index of all peers are same, then the amount of resources contributed will be same as what is taken from the network for each peer.\\
2) If resources contributed and resources taken from the network in each peer are   same, then the biased contribution index of all peers will be same.\par 
First point is evident from  section \ref{solution}. That is, if biased contribution index of all peers are same then it will be $(1-\alpha/2)$ (see Lemma \ref{lemma2}). And if biased contribution index will be $(1-\alpha/2)$ then $\mathbf{(s^T-s)e_i^{T}} \perp \mathbf{e} $  for all $i$ (from Lemma \ref{lemma3}).  Hence \[\mathbf{e_i.(s-s^T).e=0} \hspace{5mm} \forall i\] \[\Rightarrow \hspace{2mm}\mathbf{e_i.s.e=e_i.s^T.e} \hspace{5mm} \forall i\] Hence total amount of resources contributed and taken from the network by each peer will be same.\par 
Second point can be understood directly from Lemma \ref{lemma4}. If resources contributed and taken from the network in each peer are  same then
\[\mathbf{e_i.s.e=e_i.s^T.e} \hspace{5mm} \forall i\]
Hence  $\mathbf{se=s^Te}$ and $\mathbf{x}=(1-\alpha/2)\mathbf{e}.$ 
Hence biased contribution index of all peers will be same.
\subsection{Implementation In Distributed System}
Calculation of biased contribution index can be implemented in a distributed system following the same approach as in  \cite{abs} \cite{eigen} \cite{global}. Multiple other  peers named  index managers, can be assigned to maintain the record of biased contribution index of any peer. Whenever any peer needs the biased contribution index, of other peers it can send the query to the respective index managers. If there is any conflict about the biased contribution index of a peer, it can be settled, by majority of voting by index managers. In this way we can avoid the collusion among peers. For calculation of biased contribution index of any peer, index manager needs to know the contribution and resource taken by that peer and biased contribution index of peer with whom it is transacting. Calculation is repeated till the convergence of biased contribution index is achieved. If number of iterations required  to converge the algorithm are less, then required number of update messages will also be less. Therefore the  algorithm can be implemented with less overhead. In next section, we will compare the speed of convergence of our method with the other algorithm \cite{global}.

\section{Numerical Result}\label{result}
We considered the share matrix $\mathbf{s}$ as 
$$\quad
\begin{bmatrix} 
0 & 100&50&20 \\
20&0&30&40\\
10&40&0&50\\
50&10&60&0
 
\end{bmatrix}
\quad.$$ 

\begin{table}
\begin{center}
\caption{Biased Global Contribution index for, $\alpha=0.8$  in each iteration}\label{table1}
\begin{tabu} to 0.4\textwidth { | X[c] | X[c] | X[c] | X[c] | X[c]| X[c] | } 
 \hline
   $i$& 1& 2& 3& 4\\[1ex] 
 \hline
 $\mathbf{ x^0}$& 0.6000  &  0.6000 &   0.6000  &  0.6000\\
    \hline
$\mathbf{ x^1}$& 0.7440    & 0.5000  &  0.5333  &  0.6174\\
    \hline
 $\mathbf{ x^2}$&0.7266    & 0.4823  &  0.5161  &  0.6373\\
    \hline
 $\mathbf{ x^3}$&0.7202    & 0.4861  &  0.5170  &  0.6379\\
    \hline
 $\mathbf{ x^4}$&0.7207    & 0.4870  &  0.5177  &  0.6371\\
    \hline
 $\mathbf{ x^5}$&0.7210    & 0.4869  &  0.5177  &  0.6370\\
    \hline
 $\mathbf{ x^6}$&0.7210    & 0.4868  &  0.5177  &  0.6370\\
    \hline
 $\mathbf{ x^7}$&0.7210    & 0.4868  &  0.5177  &  0.6370\\

  \hline
\end{tabu}
\end{center}
\end{table}

\begin{table}
\begin{center}
\caption{Biased Global Contribution index for, $\alpha=0.4$  in each iteration}\label{table2}
\begin{tabu} to 0.4\textwidth { | X[c] | X[c] | X[c] | X[c] | X[c]| X[c] | } 
 \hline
  $i$  & 1& 2& 3& 4\\[1ex] 
 \hline
 $\mathbf{ x^0}$& 0.8000 &   0.8000 &    0.8000 &   0.8000\\
    \hline
$\mathbf{ x^1}$& 0.8720   & 0.7500  &  0.7667   & 0.8087\\
    \hline
 $\mathbf{ x^2}$&0.8690   & 0.7465   & 0.7634   & 0.8124\\
    \hline
 $\mathbf{ x^3}$&0.8685   & 0.7468   & 0.7634   & 0.8124\\
    \hline
 $\mathbf{ x^4}$&0.8686   & 0.7468   & 0.7635   & 0.8124\\
    \hline
 $\mathbf{ x^5}$&0.8686   & 0.7468   & 0.7635   & 0.8124\\
   
  \hline
\end{tabu}
\end{center}
\end{table}

\begin{table}
\begin{center}
\caption{Number of iterations, required to converge the BCI and GC \cite{global} for different values of $\alpha$ and $\beta$(parameter $\beta$ is defined as in \cite{global})  }\label{table3}
%\begin{tabu} to.5\textwidth { | X[c] | X[c] | X[c] | } 
\begin{tabular}{ | m{1cm} | m{12em}| m{8em} | } 
 \hline
 & Number of iteration required in GC\cite{global}& Number of iteration required in BCI \\[1ex] 
 \hline
  &  $\beta=0.8$, $Iterations =9$ &   \\
    \cline{2-2}
$\alpha=0.9 $& $\beta=0.5$, $Iterations=10$  &   $Iterations =8$ \\
    \cline{2-2}
  &$\beta=0.2$, $Iterations=10$  &    \\
    \hline
&  $\beta=0.8$, $Iterations=10$ &   \\
    \cline{2-2}
$\alpha=0.8 $& $\beta=0.5$, $Iterations=8$  &   $Iterations =7$ \\
    \cline{2-2}
  &$\beta=0.2$, $Iterations=11$  &    \\
    \hline
    &  $\beta=0.8$, $Iterations=9$ &   \\
    \cline{2-2}
$\alpha=0.7 $& $\beta=0.5$, $Iterations=8$  &   $Iterations =7$ \\
    \cline{2-2}
  &$\beta=0.2$, $Iterations=9$  &    \\
    \hline
    &  $\beta=0.8$, $Iterations=8$ &   \\
    \cline{2-2}
$\alpha=0.6 $& $\beta=0.5$, $Iterations=7$  &   $Iterations =6$ \\
    \cline{2-2}
  &$\beta=0.2$, $Iterations=8$  &    \\
    \hline
    &  $\beta=0.8$, $Iterations=7$ &   \\
    \cline{2-2}
$\alpha=0.5 $& $\beta=0.5$, $Iterations=6$  &   $Iterations =5$ \\
    \cline{2-2}
  &$\beta=0.2$, $Iterations=8$  &    \\
    \hline
    &  $\beta=0.8$, $Iterations=7$ &   \\
    \cline{2-2}
$\alpha=0.4 $& $\beta=0.5$, $Iterations=6$  &   $Iterations =5$ \\
    \cline{2-2}
  &$\beta=0.2$, $Iterations=7$  &    \\
    \hline
    &  $\beta=0.8$, $Iterations=6$ &   \\
    \cline{2-2}
$\alpha=0.3 $& $\beta=0.5$, $Iterations=5$  &   $Iterations =4$ \\
    \cline{2-2}
  &$\beta=0.2$, $Iterations=6$  &    \\
    \hline
    &  $\beta=0.8$, $Iterations=5$ &   \\
    \cline{2-2}
$\alpha=0.2 $& $\beta=0.5$, $Iterations=5$  &   $Iterations =3$ \\
    \cline{2-2}
  &$\beta=0.2$, $Iterations=6$  &    \\
    \hline
%\end{tabu}
\end{tabular}
\end{center}
\end{table}

\subsection{Speed of Convergence}
The number of iterations required for convergence of the biased contribution index were estimated for two different values of $\alpha$. For $\alpha= 0.8$, BCI in each step is shown in Table \ref{table1}. We can see that it converges in seven iterations. For $\alpha=0.4$ (see Table \ref{table2}), it converge only in five iterations. Thus impact of $\alpha$ is clearly evident in the results.
\subsection{Comparison of Convergence Speed With GC\cite{global}}
We compared the number of iterations required for convergence of the biased contribution index and global contribution \cite{global}. In latter case, we have taken the  different values of $\alpha$ and $\beta$ as defined in \cite{global}. Results are shown in Table \ref{table3}. We can observe that the number of iterations required for convergence of BCI is always lesser than the global contribution mentioned in \cite{global}. 
%\subsection{Convergence of the Global Trust}

\section{Conclusion}\label{con}
In this work, we propose a new metric,  the biased contribution index, to evaluate the  contributions of the peers in the network. Using this metric we can discourage  the free riding in the network. We can also ensure the balance between the total upload and download by a node  in the network. We compared our method with another existing approach \cite{global}. With the help of numerical example, we have shown that our metric converges in lesser number of iterations compared to the global contribution approach given in \cite{global}. Our approach can also be implemented in a distributed system and is much simpler than the other existing approach.

% or


%\appendices

%\section{}

%\section*{Acknowledgment}


\ifCLASSOPTIONcaptionsoff
  \newpage
\fi

\end{document}